\documentclass[conference]{IEEEtran}
\IEEEoverridecommandlockouts
\usepackage{cite}
\usepackage{amsmath,amssymb,amsfonts}
\usepackage{algorithmic}
\usepackage{algorithm}
\usepackage{graphicx}
\usepackage{textcomp}
\usepackage{subfigure}
\usepackage{array}
\usepackage{bm}
\usepackage{multirow} 
\allowdisplaybreaks[4]
\usepackage{xcolor}
\def\BibTeX{{\rm B\kern-.05em{\sc i\kern-.025em b}\kern-.08em
    T\kern-.1667em\lower.7ex\hbox{E}\kern-.125emX}}
\usepackage{ntheorem}
\theorembodyfont{\upshape}
\theoremheaderfont{\it}
\qedsymbol{\square}
\theoremseparator{:}
\newtheorem{theorem}{\indent Theorem}
\newtheorem{lemma}{\indent Lemma}
\newtheorem*{proof}{\indent Proof}
\newtheorem{remark}{\indent Remark}
\newtheorem{corollary}{\indent Corollary}

\begin{document}

\title{Quantized Analog Beamforming Enabled Multi-task Federated Learning Over-the-air
 %\thanks{This work was supported in part by National Key R\&D Program of China (Grant No. 2023YFB2904804), the NSFC under grants 62022026~and 62211530108, and the Fundamental Research Funds for the Central~Universities 2242022k60002 and 2242023K5003. (Corresponding author: Wei~Xu)}
}
\vspace{-0.5cm}
\author{\IEEEauthorblockN{Jiacheng~Yao\IEEEauthorrefmark{1}\IEEEauthorrefmark{2},
Wei~Xu\IEEEauthorrefmark{1}\IEEEauthorrefmark{2},
Guangxu~Zhu\IEEEauthorrefmark{3},
Zhaohui~Yang\IEEEauthorrefmark{4},
Kaibin~Huang\IEEEauthorrefmark{5}, and
Dusit~Niyato\IEEEauthorrefmark{6}}
\IEEEauthorblockA{\IEEEauthorrefmark{1}National Mobile Communications Research Laboratory, Southeast University, Nanjing, China}
\IEEEauthorblockA{\IEEEauthorrefmark{2}Purple Mountain Laboratories, Nanjing, China}
\IEEEauthorblockA{\IEEEauthorrefmark{3}Shenzhen Research Institute of Big Data, Shenzhen, China}
\IEEEauthorblockA{\IEEEauthorrefmark{4}College of Information Science and Electronic Engineering, Zhejiang University, Hangzhou, China}
\IEEEauthorblockA{\IEEEauthorrefmark{5}Department of Electrical and Electronic Engineering, The University of Hong Kong, Hong Kong SAR, China}
\IEEEauthorblockA{\IEEEauthorrefmark{6}School of Computer Science and Engineering, Nanyang Technological University, Singapore}
\IEEEauthorblockA{Emails: \{jcyao,wxu\}@seu.edu.cn, gxzhu@sribd.cn, yang\_zhaohui@zju.edu.cn}
\vspace{-0.7cm}}

\maketitle

\begin{abstract}
Over-the-air computation (AirComp) has recently emerged as a pivotal technique for communication-efficient federated learning (FL) in resource-constrained wireless networks. Though AirComp leverages the superposition property of multiple access channels for computation, it inherently limits its ability to manage inter-task interference in multi-task computing. In this paper, we propose a quantized analog beamforming scheme at the receiver to enable simultaneous multi-task FL. Specifically, inspiring by the favorable propagation and channel hardening properties of large-scale antenna arrays, a targeted analog beamforming method in closed form is proposed for statistical interference elimination. Analytical results reveal that the interference power vanishes by an order of $\mathcal{O}\left(1/N_r\right)$ with the number of analog phase shifters, $N_r$, irrespective of their quantization precision. Numerical results demonstrate the effectiveness of the proposed analog beamforming method and show that the performance upper bound of ideal learning without errors can be achieved by increasing the number of low-precision analog phase shifters.
\end{abstract}


\section{Introduction}
Over-the-air computation (AirComp) becomes a promising technique for enabling communication-efficient federated learning (FL) over resource-constrained wireless networks \cite{xw,amiri,tsp}. Specifically, by exploiting the superposition property of MAC channels, the signals are amplitude-modulated and simultaneously transmitted with the same radio resource and hence the summation is automatically achieved over-the-air. Therefore, applying AirComp to FL can deeply integrate uplink gradient transmission and model aggregation, breaking through the bottleneck of limited bandwidth~\cite{twcyao}.

%Federated learning (FL) is viewed as a pivotal technique for enabling distributed deployment of artificial intelligence (AI) models over wireless networks to advance inclusive intelligent services \cite{xw,pushing}. However, limited communication capabilities become a decisive factor restricting the FL performance . In particular, limited spectrum resource can hardly support low-latency and high-reliability access of massive devices, making the FL training procedure susceptible to perturbations from transmission errors. Aligning with the requirements of gradient uploading and model aggregation in FL, over-the-air computation (AirComp) becomes a promising technique. 

However, existing AirComp-enabled FL (AirFL) designs primarily focused on single-antenna systems, limiting their capability to handle upto one single computational task at a receiver. The design of multi-antenna systems mainly focused on improving diversity gain through beamforming, allowing more devices to access and minimizing the mean squared error (MSE) of gradient calculation \cite{bf0}. The potential multiplexing capability of multi-antenna systems to support multi-task computing has rarely been considered in recent works. However, emerging practical applications of FL urgently necessitate the implementation of multi-task computing, primarily including the personalized FL scenarios \cite{cl1,shi2024empowering}.

%following two scenarios. First, to address data heterogeneity in real-life scenarios with non-independent and identically distributed (non-i.i.d.) local datasets, it is necessary to train multiple personalized models in parallel to adapt to diverse data distributions and user demands \cite{cl1}. Second, in response to diverse intelligent task requirements, multiple unrelated FL tasks may simultaneously exist in  systems, requiring the PS to orchestrate parallel FL training processes.

To better support multi-task FL, it is necessary to explore the implementation of multi-task AirComp based on typical multiple-input multiple-output (MIMO) systems. However, the presence of additional inter-task interference becomes a decisive factor in determining the computational performance.
Moreover, due to the nonorthogonal multiple access characteristics of AirComp, it is challenging to directly deploy the classical zero-forcing (ZF) receiver to eliminate interference. To facilitate the implementation of multi-task AirComp, various techniques have been proposed in recent years. For instance, a large number of antennas at the transceivers are considered in \cite{icassp2}, together with coordinated ZF beamforming for complete interference elimination. To reduce the number of antennas, the authors in \cite{ssguo} combined the beamforming optimization with device scheduling, which helps reduce the number of devices and ensure interference-free computation. Without complete interference elimination, the spatial correlation between devices was considered in \cite{xjyuan}, which allows the proposition of joint transceiver beamforming and device selection design for task-oriented interference suppression.

Nevertheless, deploying a large number of antennas to tackle with inter-task interference results in significant hardware cost. Meanwhile, the training process of FL tasks have inherent robustness to noise \cite{pca}, reducing the necessity for complete interference elimination. Motivated by this background, in this paper, we adopt a low-cost hybrid beamforming architecture with limited radio frequency (RF) chains as a solution \cite{sub},  and propose an analog beamforming method in closed form to enable multiple FL tasks simultaneously without the utilization of channel state information at the transmitter (CSIT). Our theoretical analysis validates the vanishing power of inter-task interference with order $\mathcal{O}\left(1/N_r\right)$ with the number of analog phase shifters, $N_r$, which is irrespective of their quantization precision. The numerical results reveal that by employing the proposed analog beamforming method, all FL tasks can approach the upper bound performance of ideal learning with error-free transmission.

% \vspace{-15.pt}
\section{System Model}
In this paper, we consider a wireless FL system with $N$ individual FL tasks.  Without loss of generality, we assume that $K$ distributed devices are uniformly divided into $N$ clusters, yielding each cluster of $L=\frac{K}{N}$ users without overlap. A central parameter server (PS) is deployed to coordinate the parallel training of $N$ FL tasks. The training procedure and transmission model are elaborated in the sequel.

\subsection{Federated Learning Model}
%With the orchestration of the central PS, devices in each cluster collaborate  to train a shared AI model by exploiting their own local datasets. 
Let $\mathcal{D}_{n,l}$, $\forall n\in[N], l\in[L]$, denote the local dataset of the $l$-th device in the $n$-th FL task and $D_{n,l}$ represents the number of training samples in dataset $\mathcal{D}_{n,l}$. Then,  the local loss function of the $l$-th device with the $n$-th FL task is given by $f_{n,l} (\mathbf{v}_n,\mathcal{D}_{n,l})= \frac{1}{D_{n,l}} \sum_{\mathbf{u}\in \mathcal{D}_{n,l}}	\mathcal{L}(\mathbf{v}_n,\mathbf{u})$,
where $\mathbf{u}$ is the data sample selected from the local dataset and $\mathcal{L}(\mathbf{v}_n,\mathbf{u})$ represents the sample-wise loss function quantifying the prediction
error of model parameter $\mathbf{v}_n$ on data sample $\mathbf{u}$.
The learning process aims to optimize the specific model parameter $\mathbf{v}_n\in\mathbb{R}^{d}$ for the $n$-th task by minimizing the global loss function defined as
 $F_n(\mathbf{v}_n)=\sum_{l=1}^L \alpha_{n,l} f_{n,l}(\mathbf{v}_n,\mathcal{D}_{n,l})$,
where $\alpha_{n,l}\triangleq \frac{D_{n,l}}{\sum_{\ell=1}^L D_{n,\ell}}$ represents the aggregation weight for the $l$-th device. For simplicity, we assume that for all tasks, the models contain the same number of parameters, denoted by $d$. For models with different sizes, zero padding can be employed to ensure a uniform size, thereby applicable to general scenarios.

For model training, we employ the typical FedAvg algorithm to iteratively minimize $F_n(\mathbf{v}_n)$. The main steps of the $t$-th round of the FedAvg algorithm are listed as follows.
\begin{itemize}
\item [1)]\emph{Model broadcasting}: The PS broadcasts the latest global model $\mathbf{v}_{n,t}$ to all devices with the $n$-th FL task. 
\item [2)]\emph{Local computing}: Based on the received global model $\mathbf{v}_{n,t}$, each device updates its local model via mini-batch stochastic gradient descent (SGD), i.e., $\mathbf{v}_{n,t,i+1}^l= \mathbf{v}_{n,t,i}^l-\eta_{n,t}\nabla f_{n,l}\left(\mathbf{v}_{n,t,i}^l, \bm{\xi}_{n,t,i}^l\right)$,
where $\mathbf{v}_{n,t,i}^l$ denote the local model of the $l$-th device obtained after the $i$-th local iteration, $\mathbf{v}_{n,t,0}^l$ is initialized as $\mathbf{v}_{n,t}$, $\eta_{n,t}$ denotes the learning rate, and $\bm{\xi}_{n,t,i}^l$ is a local mini-batch selected from $\mathcal{D}_{n,l}$. Each device performs $E$ local iterations in a specific round and updates its local model as $\mathbf{v}_{n,t,E}^l$.
\item [3)]\emph{Local updates uploading}: After local computing, each device transfers its local 
updates $\mathbf{g}_{n,l,t}$ to the PS, where  $\mathbf{g}_{n,l,t}\!\triangleq\!\frac{\mathbf{v}_{n,t,0}^l\!-\!\mathbf{v}_{n,t,E}^l}{\eta_{n,t}}\!=\!\sum_{i=0}^{E-1}\! \nabla f_{n,l}\!\left(\mathbf{v}_{n,t,i}^l, \bm{\xi}_{n,t,i}^l\right)$. 
\item [4)]\emph{Model aggregation}: Upon receiving all the local updates, the PS calculates the global update as
$\mathbf{g}_{n,t} = \sum_{l=1}^L \alpha_{n,l} \mathbf{g}_{n,l,t}$,
and then updates the global model according to $\mathbf{v}_{n,t+1}=\mathbf{v}_{n,t}-\eta_{n,t}\mathbf{g}_{n,t}$.
\end{itemize}

\subsection{Multi-task AirComp Model}

The downlink model broadcasting is usually assumed error-free due to abundant resource at the PS \cite{twcyao}. For the uplink, we adopt the AirComp technique for simultaneous transmission and model aggregation. Moreover, to enable parallel model aggregation for multiple FL tasks,  the PS exploits hybrid analog and digital beamforming with a typical sub-connected structure \cite{sub}. In particular, the PS is equipped with $N_{\text{RF}}$ RF chains and $N_r$ antennas. Each RF chain connects to an exclusive set of $M$ antennas through a dedicated phase shifter, where $M=\frac{N_r}{N_{\text{RF}}}$.  We consider the scenario with $N_{\text{RF}} = N$, where each RF chain and its connected analog array is dedicated to serving a specific FL task.

Now, it is ready to introduce the considered multi-task AirComp. Before transmission, local updates are first normalized~as $\mathbf{s}_{n,l,t}= \frac{\mathbf{g}_{n,l,t}-\mu_{n,l,t}\mathbf{1}}{v_{n,l,t}}$,
where $\mu_{n,l,t}$ and $v_{n,l,t}$ respectively represent the mean and standard deviation of all entries in $\mathbf{g}_{n,l,t}$, and $\mathbf{1}$ is all one vector.
Unlike traditional AirComp relying on perfect CSIT, we assume that no CSIT is available at the device, and consequently corresponding preprocessing, such as the typical channel inversion scheme in \cite{imperfect}, is not performed. Hence, the received signal at the PS, $\mathbf{Y}_t\in\mathbb{C}^{N\times d}$, is given as follows:
\begin{align} \label{receive}
    \mathbf{Y}_t=\mathbf{W}_t\mathbf{A}_t \sum_{i=1}^N \sum_{l=1}^L \sqrt{p_{i,l,t}}\mathbf{h}_{i,l,t}\mathbf{s}_{i,l,t}^T+\mathbf{W}_t\mathbf{A}_t \mathbf{N}_t,
\end{align}
where $\mathbf{h}_{i,l,t}\sim \mathcal{CN}\left(\mathbf{0}_{N_r\times 1},\mathbf{I}_{N_r}\right)$ denotes the channel between the PS and the $l$-th device in the $i$-th cluster, $p_{i,l,t}$ represents its transmit power, $\mathbf{W}_t=[\mathbf{w}_{1,t},\mathbf{w}_{2,t},\cdots,\mathbf{w}_{N,t}]^H \in \mathbb{C}^{N\times N}$ and $\mathbf{A}_t\in \mathbb{C}^{N\times N_r}$ represent the digital combiner and analog beamformer, respectively, and $\mathbf{N}_t\in \mathbb{C}^{N\times d}$ denotes the additive Gaussian noise matrix and its entries are independent and identically Gaussian distributed with zero mean and variance of $\sigma^2$. With the sub-connected structure, the analog beamformer can be expressed as $\mathbf{A}_t=\text{diag}\left\{ \mathbf{a}_{1,t}^H,\cdots, \mathbf{a}_{N,t}^H\right\}$,
where $\mathbf{a}_{i,t}$ denotes the analog beamforming vector of the $i$-th sub-array satisfying $\left | [\mathbf{a}_{i,t}]_j\right|^2=\frac{1}{M}$, $\forall j\in[M]$. Since the $n$-th RF chain is dedicated to serving the $n$-th FL task, we exploit the processed signal at the $n$-th RF chain as an estimate of the desired weighted aggregation, $\sum_{l=1}^L \alpha_{n,l} v_{n,l,t} \mathbf{s}_{n,l,t}$. It follows
\begin{align} \label{process}
    \mathbf{y}_{n,t}^T = \mathbf{w}_{n,t}^H \sum_{i=1}^N \sum_{l=1}^L \sqrt{p_{i,l,t}}\mathbf{A}_t\mathbf{h}_{i,l,t}\mathbf{s}_{i,l,t}^T+\mathbf{w}_{n,t}^H \mathbf{A}_t\mathbf{N}_t.
\end{align}
The global update for the $n$-th task is then calculated as
\begin{align}\label{ghat}
    \hat{\mathbf{g}}_{n,t}=\mathbf{y}_{n,t}+\sum_{l=1}^L \alpha_{n,l}\mu_{n,l,t}\mathbf{1},
\end{align}
where the statistics of the local updates are transmitted to PS in advance with ignorable communication overhead \cite{nfzhang}. By comparing (\ref{ghat}) with the ideal local updates, we express the error of global update estimation, $\mathbf{e}_{n,t}\triangleq \hat{\mathbf{g}}_{n,t}- \mathbf{g}_{n,t}$, as
\begin{align}\label{error}
    \mathbf{e}_{n,t}
    =& \sum_{l=1}^L \left(\sqrt{p_{n,l,t}}\mathbf{w}_{n,t}^H\mathbf{A}_t\mathbf{h}_{n,l,t}-\alpha_{n,l}v_{n,l,t}\right)\mathbf{s}_{n,l,t}\nonumber \\
    &+ \sum_{i\neq n}^N \sum_{l=1}^L \sqrt{p_{i,l,t}} \mathbf{w}_{n,t}^H \mathbf{A}_t\mathbf{h}_{i,l,t}\mathbf{s}_{i,l,t}\!+\! \mathbf{N}_t^T\mathbf{A}_t^T\mathbf{w}_{n,t}^*.
\end{align}
Note that the error in (\ref{error}) comprises of two main components. First, there is distortion in the aggregation weights of the desired signal. The second component consists mainly of interference from other tasks and the additive Gaussian noise.  To tackle with the inter-task interference, existing works \cite{icassp2,ssguo} rely on sufficient spatial degrees of freedom (DoF) to apply ZF receivers for complete interference elimination, which requires a large number of  RF chains, i.e., $N_{\text{RF}}>K$. However, in the scenario considered in this paper, the number of RF chains is limited, making complete interference elimination unattainable. In the subsequent section, we introduce how analog beamforming design is utilized to statistically eliminate interference and support the multi-task FL over-the-air. 

\section{Statistical Interference Elimination via Analog Beamforming}

Unlike traditional communication tasks, the introduction of zero-mean noise does not necessarily have a serious impact on the SGD process while it can sometimes even impose a positive effect \cite{sgd1}.
However, noise with non-zero mean inevitably introduces bias to the SGD process, leading to performance deterioration and potential failure of training convergence. These observations encourage us to avoid unnecessary resource overhead for complete interference elimination. Instead, we rely on low-cost analog beamforming to statistically eliminate the impact of inter-task interference.

%Accordingly, the proposed analog beamforming scheme and the corresponding performance analysis are provided in this section.

\subsection{Typical Fully-Digital Beamforming Scheme}
\label{secRO}
Favorable propagation is a unique characteristic of large-scale MIMO systems. When the number of antennas, $N_r\to \infty$, the channels of different devices are asymptotically orthogonal, i.e.,  $\mathbf{h}_{i,l,t}^H \mathbf{h}_{i^\prime,l^\prime,t}\to 0$, $\forall (i,l)\neq (i^\prime,l^\prime)$. Inspired by the nature of favorable propagation, the inter-task interference can be statistically eliminated with fully-digital beamforming  at the receiver by applying a linear projection \cite{ro}, referred to as random orthogonalization (RO). To be specific, we set the fully-digital beamformer for the $n$-th task as $\mathbf{f}_{n,t}=\sum_{l=1}^L \mathbf{h}_{n,l,t}\in \mathbb{C}^{N_r}$ and the received signal in (\ref{process}) becomes 
\begin{align}\label{eq13}
    \mathbf{y}_{n,t}^T &= \sum_{i=1}^N \sum_{l=1}^L \sqrt{p_{i,l,t}}\mathbf{f}_{n,t}^H\mathbf{h}_{i,l,t}\mathbf{s}_{i,l,t}^T+\mathbf{f}_{n,t}^H \mathbf{N}_t\nonumber \\
    &\overset{\text{a.s}}{\to} \sum_{l=1}^L \sqrt{p_{n,l,t}}\left \Vert \mathbf{h}_{n,l,t}\right \Vert^2 \mathbf{s}_{i,l,t}^T,
\end{align}
where $\overset{\text{a.s}}{\to}$ represents ``almost surely converge to" and it is due to the favorable propagation property with $N_r\to \infty$, and the inter-device interference is asymptotically eliminated. Also, from the statistical perspective, we have
\begin{align}
    \mathbb{E}\left[ \mathbf{h}_{i,l,t}^H \mathbf{h}_{i^\prime,l^\prime,t}
    \right] =0, \enspace \forall (i,l)\neq (i^\prime,l^\prime).
\end{align}
Hence, it is straightforward to verify that the inter-task interference is statistically eliminated.
The receive beamforming acts as a filter, which filters out interference components orthogonal to it. However, achieving better interference elimination requires a substantial deployment of extra RF chains, significantly increasing hardware cost. Alternatively, large-scale antenna arrays enable the direct use of ZF receivers to obtain individual signals from different devices without interference. The RO approach may not necessarily yield better performance compared to the ZF receiver.

\subsection{Proposed Analog Beamforming with Continuous Phases}

The core idea of (\ref{eq13}) is to include the target channel components, i.e., $\mathbf{h}_{n,l,t}$, $\forall l$, in the $n$-th digital beamformer, thus achieving statistical interference elimination via the favorable propagation property. However, the included amplitude components of $\mathbf{h}_{n,l,t}$, $\forall l$, do not exert a significant impact. Hence, we posit that utilizing analog phase shifters alone can still achieve this statistical interference elimination. In particular, considering that the $n$-th sub-array only serves the $n$-th FL task, we set its analog beamforming as 
\begin{align}\label{analog}
    \mathbf{a}_{n,t}=\frac{1}{\sqrt{M}}\mathrm{exp}\left( j\angle \sum_{l=1}^L \mathbf{h}_{n,l,t,n}\right),\enspace\forall n
\end{align}
where $\mathbf{h}_{n,l,t,n}\in \mathbb{C}^{M\time 1}$ denotes the channel between the $n$-th sub-array and the $l$-th device in the $n$-th cluster in round $t$ and $\mathbf{h}_{n,l,t}=\left[\mathbf{h}_{n,l,t,1}^H,\cdots,\mathbf{h}_{n,l,t,N}^H \right]^H$, $\forall n\in[N], l\in[L]$.  In (\ref{analog}), we only incorporate the phases of channels related to devices in the $n$-th cluster, aiming at filtering out interference from other clusters. With the analog beamforming in (\ref{analog}), the effective channel for the $l$-th device in the $n$-th cluster at the $t$-th round is given as $\bar{\mathbf{h}}_{n,l,t}=\mathbf{A}_{t}\mathbf{h}_{n,l,t}$, whose statistics are derived in the following lemma.

\begin{lemma} \label{theorem1}
    With the analog beamforming in (\ref{analog}), the effective channel follows
    \begin{align}
        &\mathbb{E}\left[ \bar{\mathbf{h}}_{n,l,t}\right]=\frac{\sqrt{\pi M}}{2\sqrt{L}} \bm{\delta}_n,\enspace \mathbb{V}\left[ \bar{\mathbf{h}}_{n,l,t}\right]=\bm{\Sigma}_n,\enspace\forall n,
    \end{align}
    where $\bm{\delta}_n$ is the Kronecker delta vector with $[\bm{\delta}_n]_n=1$, $\mathbf{\Sigma}_n$ is a diagonal matrix, and its $n$-th diagonal element is $1-\frac{\pi}{4L}$ and all the other diagonal elements are 1.
\end{lemma}
\begin{proof}
    Please refer to  the supplementary file in \cite{shi2024empowering}. $\hfill\square$
\end{proof}

Building upon the analytical results in \emph{Lemma \ref{theorem1}}, the statistical interference elimination can be achieved only with analog beamforming. To be concrete, the digital combiner $\mathbf{W}_t$ does not cope with the inter-task interference. We set $\mathbf{W}_t=\zeta \mathbf{I}_N$, where $\zeta>0$ is a scaling factor. Then, it follows
\begin{align}\label{expectation}
    \mathbb{E}\left[\hat{\mathbf{g}}_{n,t}\right]&= \sum_{l=1}^L \frac{\zeta\sqrt{\pi Mp_{n,l,t}}}{2\sqrt{L}}\mathbf{s}_{n,l,t},
\end{align}
which no longer includes the inter-task interference terms.

To further demonstrate the statistical interference elimination method, we  characterize the scaling law of the average power of the interference with respect to the number of antennas, $N_r$, in the following theorem. The average power of the interference for the $n$-th task is given by
\begin{align}
    P_{\text{I},n}\triangleq \mathbb{E}\left[
    \left\Vert \sum_{i\neq n}^N \sum_{l=1}^L \sqrt{p_{i,l,t}} \mathbf{w}_{n,t}^H \mathbf{A}_t\mathbf{h}_{i,l,t}\mathbf{s}_{i,l,t} \right \Vert^2
    \right].
\end{align}

\begin{theorem} \label{lemma2}
    As the number of antennas, $N_r$, increases, the average power of the interference $P_{\text{I},n}$, $\forall n$, vanishes  by an order of~$\frac{1}{N_r}$.
\end{theorem}

\begin{proof}
    From (\ref{error}), the average power of the interference for the $n$-th task is expressed as
    \begin{align}\label{eq20}
    &P_{\text{I},n}\!\overset{\text{(a)}}{=}d\sum_{i\neq n}^N \sum_{l=1}^L p_{i,l,t}\mathbb{E}\left[
    \left\vert \mathbf{w}_{n,t}^H \mathbf{A}_t\mathbf{h}_{i,l,t}\right \vert^2
    \right]\nonumber \\
    &\!=\!d\zeta^2\! \!\sum_{i\neq n}^N \!\sum_{l=1}^L \!p_{i,l,t}\bm{\delta}_{n}^H\mathbb{E}\!\!\left[
    \bar{\mathbf{h}}_{i,l,t} \bar{\mathbf{h}}_{i,l,t}^H
    \right]\!\bm{\delta}_{n}\!\!=\!d\zeta^2 \!\!\sum_{i\neq n}^N \!\sum_{l=1}^L \!p_{i,l,t},
    \end{align}
    where (a) is due to the independence between signal vectors and the fact that $\mathbb{E}\left[\mathbf{s}_{i,l,t}\right]=\mathbf{0}$. Moreover, as observed in (\ref{expectation}), the coefficient of the desired signal $\mathbf{s}_{n,l,t}$ is $\frac{\zeta\sqrt{\pi N_r p_{n,l,t}}}{2\sqrt{K}}$, which scales up with $\sqrt{N_r}$. To ensure a fixed aggregation coefficient, the scaling factor $\zeta$ must satisfy  $\zeta \propto \frac{1}{\sqrt{N_r}}$.
    Now, we can conclude that $P_{\text{I},n}\propto \frac{1}{N_r}$ and the proof completes.
    \hfill $\square$
\end{proof}

\begin{remark}
    By increasing the numbers of low-cost phase shifters, the proposed analog beamforming method can effectively mitigate interference without the need for additional RF chains. Moreover, with infinite phase shifters, i.e., $N_r \to \infty$, we have $P_{\text{I},n}\to0$ and the inter-task interference is asymptotically completely eliminated. 
\end{remark}

\begin{corollary}
    The proposed analog beamforming method enjoys the same scaling law for interference power with the RO scheme using fully-digital beamforming in Sec. \ref{secRO}.
\end{corollary}

\begin{proof}
    It has been demonstrated in \cite[Eq. (10)]{ro} that  interference decrease by an order of $\frac{1}{N_r}$, which is same with the conclusion in \emph{Theorem \ref{lemma2}}. \hfill $\square$
\end{proof}

\begin{remark}
    The proposed analog beamforming method fully replicates the functionality of the RO scheme. This implies that achieving random orthogonalization requires only phase matching, thereby avoiding high hardware cost associated with numerous RF chains for the fully-digital beamforming.
\end{remark}

\vspace{-0.3cm}
\subsection{Extension to Discrete Phase Control}
Due to hardware limitations, phase shifts are usually limited to a finite number of discrete values. We further extend the proposed analog beamforming design to the scenarios with discrete phase control. Specifically, we denote the set of discrete phase shifts by $\mathcal{A}\triangleq \left \{ 0,\frac{2\pi}{2^b},\cdots, \frac{(2^b-1)2\pi}{2^b}  \right \}$, where $b$ is the number of quantization bits. Then, we rewrite (\ref{analog}) as 
\begin{align} \label{analog2}
    \tilde{\mathbf{a}}_{n,t}\!=\!\frac{1}{\sqrt{M}}\mathrm{exp}\left( \!j
    \arg \!\!\min_{\bm{\phi}\in \mathcal{A}^M}\!\!\left\{\left \Vert 
    \angle \!\sum_{l=1}^L \mathbf{h}_{n,l,t,n}\!\!-\!\bm{\phi}\right \Vert^2\!\right\}\!\right)\!,
\end{align}
where $\bm{\phi}$ is an $M$-dimensional vector with elements selected from $\mathcal{A}$. Comparing with the perfect case in (\ref{analog}), the configured phase is disturbed by a quantization error, which is modelled as uniformly distributed RVs following $\mathcal{U}\left( -2^{-b}\pi, 2^{-b}\pi\right)$ \cite{sop}. Then, the discrete phase control in (\ref{analog2}) is rewritten as 
\begin{align}\label{analog3}
    \tilde{\mathbf{a}}_{n,t}=\frac{1}{\sqrt{M}}\mathrm{exp}\left( j\angle \sum_{l=1}^L \mathbf{h}_{n,l,t,n}+j\bm{\psi}_{n,t}\right),
\end{align}
where $\bm{\psi}_{n,t}$ denotes the quantization error and $[\bm{\psi}_{n,t}]_m \sim \mathcal{U}\left( -2^{-b}\pi, 2^{-b}\pi\right)$, $\forall m$. Accordingly, we denote the effective channel for the $l$-th device in the $n$-th cluster at the $t$-th round as $\tilde{\mathbf{h}}_{n,l,t}=\tilde{\mathbf{A}}_t \mathbf{h}_{n,l,t}$, where $\tilde{\mathbf{A}}_t=\text{blkdiag}\left\{\tilde{\mathbf{a}}_{1,t}^H,\cdots,\tilde{\mathbf{a}}_{N,t}^H\right\}$. %Based on the discrete phase control, we derive the following corollary

\begin{lemma}
    With the discrete phase control in (\ref{analog3}), the expectation and variance of the effective channels are
    \begin{align}
        &\mathbb{E}\left[ \tilde{\mathbf{h}}_{n,l,t}\right]\!=\!\frac{ \sin\left( 2^{-b}\pi\right)\sqrt{M}}{2^{-b+1}\sqrt{\pi L}} \bm{\delta}_n,\, \mathbb{V}\left[ \tilde{\mathbf{h}}_{n,l,t}\right]=\tilde{\bm{\Sigma}}_n,\,\forall n,
    \end{align}
    where $\tilde{\bm{\Sigma}}_n$ is a diagonal matrix, and its $n$-th diagonal element is $1-\frac{\sin^2\left(2^{-b}\pi\right)}{4^{-b+1}\pi L}$ and all the other diagonal elements are 1.
\end{lemma}

\begin{proof}
    The expectation of $e^{j[\bm{\psi}_{n,t}]_m}$ is equal to $\frac{ \sin\left( 2^{-b}\pi\right)}{2^{-b}\pi}$. Then, combining with \emph{Lemma \ref{theorem1}}, we complete the proof.
    \hfill $\square$
\end{proof}

\begin{remark}
    Compared to the results in \emph{Lemma \ref{theorem1}}, discrete phase shifts brings about only a constant scaling but does not change the statistical properties of effective channels. Therefore, the statistical interference elimination is also attained by applying the discrete phase control.
\end{remark}

Next, regarding the analysis of the average power of the interference, we derive the following lemma.

\begin{lemma}\label{lemma3}
    Compared with the continuous phase shifts, the average power of the interference obtained via discrete phase control, denoted by $P_{\text{I},n}^{\text{D}}$, satisfies $\frac{P_{\text{I},n}^{\text{D}}}{P_{\text{I},n}}=\frac{1}{\mathrm{sinc}^2\left(2^{-b}\right)}$,
    where $\mathrm{sinc}(x)\triangleq \frac{\sin(\pi x)}{\pi x}$.
\end{lemma}

\begin{proof}
    To achieve the same aggregation coefficient, we have 
    \begin{align}
        \frac{\zeta_1\sqrt{\pi Mp_{n,l,t}}}{2\sqrt{L}}=\frac{ \zeta_2 \sin\left( 2^{-b}\pi\right)\sqrt{Mp_{n,l,t}}}{2^{-b+1}\sqrt{\pi L}},
    \end{align}
    where $\zeta_1$ and $\zeta_2$ respectively denote the scaling factors for the continuous and discrete phase control. Moreover, according to (\ref{eq20}), the average power of the interference is proportional to $\zeta^2$. Hence, we complete the proof. \hfill $\square$
\end{proof}

As observed in \emph{Lemma \ref{lemma3}}, the performance gap between the continuous and discrete phase control methods vanishes with increasing $b$. Moreover, $P_{\text{I},n}^{\text{D}}$ also decreases by $\frac{1}{N_r}$ as $N_r$ increases and we conclude the following corollary.

\begin{corollary}
    Regardless of the quantization precision of phase shifts, the proposed analog beamforming method always achieves the same scaling law for interference power.
\end{corollary}
The proposed analog beamforming method, designed to prioritize favorable propagation over precise signal processing, demonstrates resilience to the impacts of non-ideal hardware. Hence, it suggests us to use a large number of low-precision phase shifters to achieve a satisfactory performance of interference elimination.

\begin{figure*}[!t]
	\centering
	\begin{minipage}[t]{0.32\linewidth}
		\centering
		\includegraphics[width=1\linewidth]{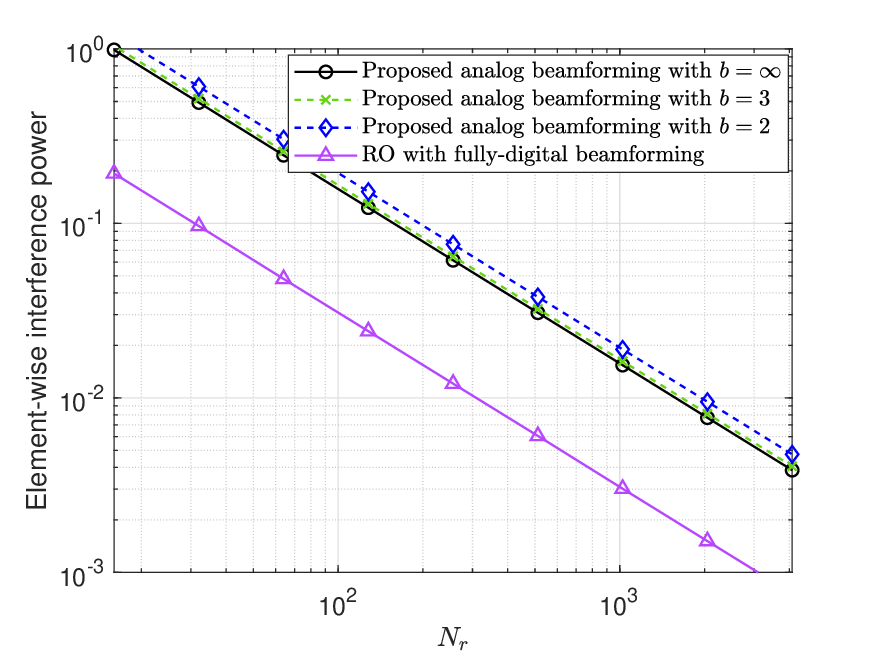}
		\caption{Element-wise power of the interference versus $N_r$ ($K=100, N=4$).}\label{fig1}
	\end{minipage}
	\begin{minipage}[t]{0.32\linewidth}
		\centering
		\includegraphics[width=1\linewidth]{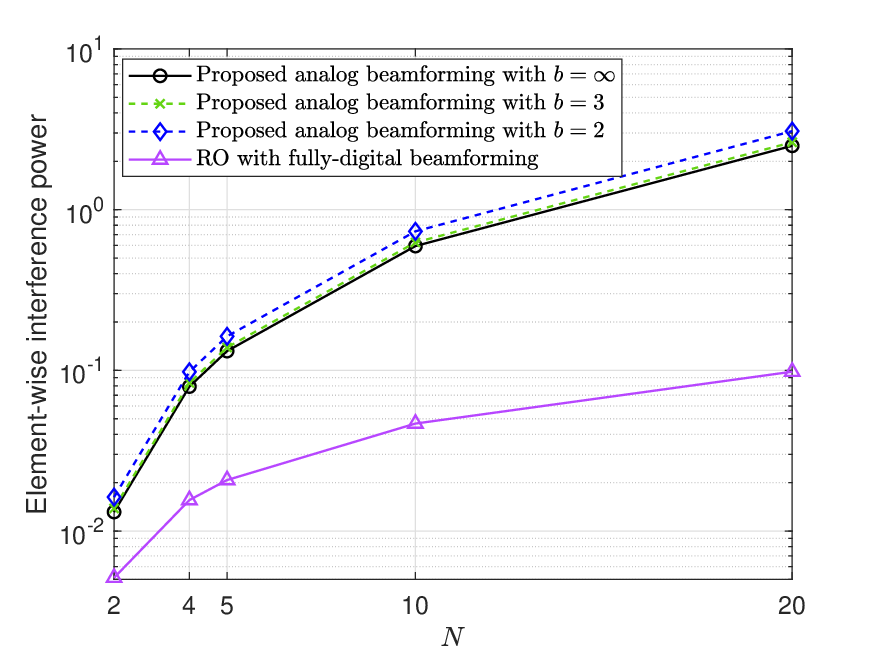}
		\caption{Element-wise power of the interference versus $N$ ($K=100, N_r=200$).}\label{fig2}
	\end{minipage}
	\begin{minipage}[t]{0.32\linewidth}
		\centering
		\includegraphics[width=1\linewidth]{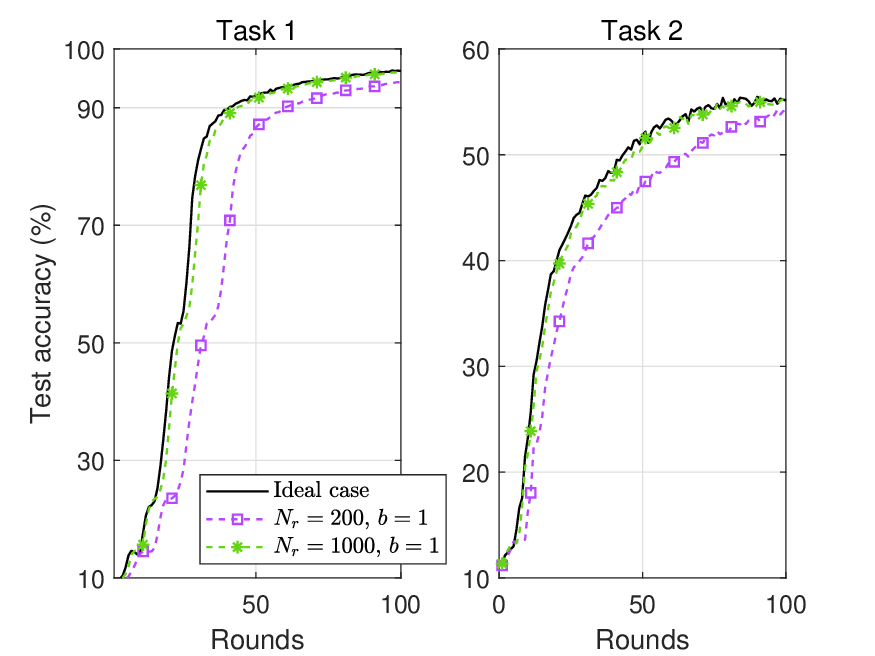}
		\caption{Test accuracy versus communication rounds ($K=20$, $N=2$, $\text{SNR}=0$ dB).}\label{fig7}
	\end{minipage}
 \vspace{-0.4cm}
\end{figure*}

\section{Numerical Results}
In this section, numerical simulations are presented to validate the proposed scheme. Without loss of generality, we neglect the large-scale path loss and normalize it as 1. Besides, we adopt the same transmit power budget of each devices, denoted by $P$. The signal-to-noise ratio (SNR) is defined as $\frac{P}{\sigma^2}$. 
To evaluate the learning performance, we simultaneously conduct two FL tasks of image classification. The two tasks are performed on the two popular datasets, i.e., MNIST and CIFAR-10, respectively. The trained AI model is a convolutional neural network (CNN).  The learning parameters are set as: the batch size $64$, the learning rate $\eta_{n,t}=0.01$, and the number of local iterations $E=1$ and $E=5$ for the FL tasks on MNIST and CIFAR-10, respectively.

Fig. \ref{fig1} depicts the element-wise power of the interference with different number of phase shifters (antennas), $N_r$. The element-wise power of the interference is defined as $\frac{\sum_{n=1}^NP_{\text{I},n}}{dN}$, representing the normalized interference power. It is observed that, in line with the theoretical results, the power of the interference decreases by an order of $\frac{1}{N_r}$ with $N_r$ using the proposed analog beamforming, and exhibits the same order with that using RO. Furthermore, we observe that using a low precision phase shifter does not result in significant performance loss. The performance comparable to that with ideal continuous phase shifters can be achieved with $b=3$. These results demonstrate that the additional expensive RF chains can be entirely replaced with low-cost phase shifters. 

In Fig. \ref{fig2}, we show the impact of the number of tasks, $N$. As $N$ increases, all schemes experience greater interference. Compared to the RO method, the performance gap between the proposed analog beamforming and it gradually  widens with increasing $N$. This is because as the number of tasks increases, the number of phase shifters assigned to each task decreases, leading to reduced interference suppression performance. This indicates that when faced with more tasks, additional analog phase shifters are required to achieve satisfactory performance.

Finally, we evaluate the convergence performance of the FL tasks in Fig. \ref{fig7}. We exploit the ideal case with perfect transmission as a benchmark. When sufficient phase shifters are deployed, e.g., $N_r=1000$, the proposed analog beamforming method effectively eliminates interference and approaches the FL performance under ideal case, even with low-precision phase shifters. 
%This fully validates the effectiveness of the proposed method in dealing with multi-task FL. 

\vspace{-0.1cm}
\section{Conclusion}
In this paper, we have developed a multi-task AirFL framework based on analog beamforming at the PS. Following the favorable propagation and channel hardening properties, we have designed a quantized analog beamforming method for statistical interference elimination. It has validated the cost-effective characteristic of the proposed method compared to existing fully-digital beamforming method based on RO.

\bibliographystyle{IEEEtran}
%\footnotesize
\bibliography{IEEEabrv,ref}

\end{document}